\documentclass[copyright]{eptcs}
\providecommand{\event}{E. Formenti (Ed.): AUTOMATA \& JAC 2012}

\usepackage[utf8]{inputenc}
\usepackage[english]{babel}
\usepackage{graphicx}
\usepackage{amsmath}
\usepackage{amssymb}
\usepackage{amsthm}
\usepackage{xcolor}
\usepackage{array}
\usepackage{pgf}
\usepackage{tikz}
\usepackage{tkz-graph}
\usetikzlibrary{automata,positioning}

\newtheorem{definition}{Definition}
\newtheorem{lemma}{Lemma}
\newtheorem{proposition}{Proposition}
\newtheorem{theorem}{Theorem}

\title{Boolean networks synchronism sensitivity and {\Large XOR} circulant
  networks convergence time\thanks{With the approval of the chairman of Automata
    2012, we include in this paper parts of another paper that has been
    submitted and is on the way to be published in a special issue of
    Theoretical Computer Science on discrete structures.}}

\author{
	Mathilde Noual
	\institute{Universit\'e de Lyon, \'ENS-Lyon, LIP, CNRS UMR 5668, 69007 Lyon, France}
	\institute{Institut rh\^one-alpin des syst\`emes complexes, IXXI, 69007 Lyon, France}
	\email{mathilde.noual@ens-lyon.fr}
	\and
	Damien Regnault
	\institute{Universit\'e d'\'Evry -- Val d'Essonne, IBISC, \'EA 4526, 91000 \'Evry, France}
	\email{damien.regnault@ibisc.univ-evry.fr}
	\and 
	Sylvain Sen\'e
	\institute{Universit\'e d'\'Evry -- Val d'Essonne, IBISC, \'EA 4526, 91000 \'Evry, France}
	\institute{Institut rh\^one-alpin des syst\`emes complexes, IXXI, 69007 Lyon, France}
	\email{sylvain.sene@ibisc.univ-evry.fr}	
}

\def\titlerunning{Boolean networks synchronism sensitivity and {\small XOR} 
circulant networks convergence time}
\def\authorrunning{M. Noual, D. Regnault \& S. Sen\'e}

\begin{document}

\maketitle

\begin{abstract}
	In this paper are presented first results of a theoretical study on the
	role of non-monotone interactions in Boolean automata networks. We
    propose to analyse the contribution of non-monotony to the diversity and
	complexity in their dynamical behaviours according to two axes. The
	first one consists in supporting the idea that non-monotony has a
	peculiar influence on the sensitivity to synchronism of such
	networks. It leads us to the second axis that presents preliminary
	results and builds an understanding of the dynamical behaviours, in
	particular concerning convergence speeds, of specific non-monotone
	Boolean automata networks called \textsc{xor} circulant
	networks.
\end{abstract}

\section{Introduction}

Boolean automata networks were first introduced by McCulloch and Pitts
in~\cite{McCulloch1943} and Kauffman in~\cite{Kauffman1969a}. These two works 
and others following these (see~\cite{Hopfield1982,Kauffman1971,Thomas1973}) 
highlighted the importance of embedding biological problematics in a context close to 
discrete mathematics and theoretical computer science.\medskip

In the lines of these studies, we propose in this paper to tackle the
question of the role of non-monotony in Boolean automata networks. Our interest
in this issue comes from the fact that non-monotony, although widely studied in
other contexts~\cite{Cull1971,Elspas1959,Huffman1956}, is missing from the
literature related to Boolean automata networks viewed as models of genetic
regulation networks.  Indeed, classically, the interaction structure of Boolean
models of genetic regulation networks are often represented by \emph{signed}
digraphs whose vertices represent genes, and arcs labelled by $+$ (resp. by $-$)
represent activations (resp.  inhibitions) of genes on each other. Thus, a gene
that tends to influence the expression of another gene is supposed to be either
one of its activators or one of its inhibitors, rather than both. More
precisely, it cannot act as an activator under certain circumstances and act as
an inhibitor under others. This interpretation of gene regulations leads to
define monotone Boolean automata networks as studied
in~\cite{Aracena2006,Cosnard1997,Goles1985,Mendoza1998,Remy2003}. Interesting
facts are that, from the general point of view, the global dynamical properties
of non-monotone networks have not yet been at the centre of studies in this
domain nor has the impact of non-monotone interactions yet been examined
\textit{per se}. It therefore seems pertinent to address questions about the
role of non-monotony on the dynamical characteristics of Boolean automata
networks. To go further, our recent theoretical developments have led us to
think that non-monotony may be at the origin of singular behaviours of these
networks. This gives additional significance to the issue of non-monotony from
both the perspectives of the theory of Boolean automata networks and of the
framework of genetic regulation networks. Thus, we present in this paper
the grounds of a larger study on non-monotony in networks by developing two
lines.  The first one consists in understanding the synchronism sensitivity of
networks. To do so, we highlight that networks can be synchronism sensitive at
different levels and shows that non-monotony is a central structural parameter
that helps to classify networks. Then, on the basis of the first line, we present
primary results on the dynamical properties (notably in terms of convergence
time) of a specific class of non-monotone networks called \textsc{xor} circulant
networks.\medskip

In Section~\ref{sec_def}, we provide definitions and notations of Boolean 
automata network theory that are used in the paper. Section~\ref{sec_synchro} 
gives details about a classification of such networks according to their synchronism 
sensitivity and show that non-monotony is a central parameter in this context. 
Then, Section~\ref{sec_xor} presents dynamical properties of \textsc{xor} circulant 
networks by exploiting their trajectories and their convergence time. Finally,
Section~\ref{sec_conclu} proposes perspectives arising from this work.

\section{Definitions and notations}
\label{sec_def}

A \emph{Boolean automata network} involves interacting elements. Any
of these elements has a state which equals $0$ or $1$. Then, we speak of
inactive and active elements respectively. Moreover, the state of each element
can change over time according to the states of other elements and to their
influence on it~\cite{Noual2011b,Robert1986}. This section is devoted to the
formalisation of the main definitions and notations used in the sequel.

\subsection{Network definition}
\label{sec_net_def}

A \emph{Boolean automata network} $N$ of size $n$ is composed of $n$
elements called \emph{automata} which are numbered from $0$ to $n-1$ such that
$V = \{0, \ldots, n-1\}$. Every automaton $i$ has a \emph{state} $x_i$ that
takes values in $\{0,1\}$. The \emph{time space is discrete} and equals
$\mathbb{N}$. The allocation of a value of $\{0,1\}$ to every automaton of $N$
is called a \emph{configuration} of $N$. It is represented by a vector $x =
(x_0, \ldots, x_{n-1}) \in \{0,1\}^n$. We also denote by $x(t)$ (resp. $x_i(t)$)
the configuration of $N$ (resp. the state of automaton $i$) at time step $t \in
\mathbb{N}$. The \emph{density} of a configuration $x$ is defined as $d(x) =
\frac{1}{n}\cdot|\{x_i\ |\ (i \in V) \land (x_i = 1)\}|$. Because we are
particularly concerned with switches of automata states starting in a given
configuration, we introduce the following notations:
\begin{multline}
	\label{eq_notations_basiques}
	\forall x = (x_0, \ldots, x_{n-1}) \in \{0,1\}^n,\\ 
	\forall i \in V = \{0, \ldots, n-1\},\ \overline{x}^i = (x_0, \ldots, x_{i-1}, \neg x_i, x_{i+1},
	\ldots, x_{n-1}) \quad \text{and} \quad \forall W\subseteq V,\ \overline{x}^{W\cup\{i\}}=
	\overline{\overline{x}^W}^i \text{.}
\end{multline}
Notice that $\overline{0}^i$ is the configuration in which $i\in V$ is the only
automaton that has state $1$ and $\overline{0}^W$ is the configuration in which
the automata in $W$ all have state $1$ contrary to automata in $V\setminus
W$. The interaction structure of $N$ is represented by a digraph $G = (V, A)$,
called the \emph{interaction graph} of $N$ that specifies what influences apply
to each automaton of $N$. In $G$, $V$ equals the set of automata of $N$ and $A
\subseteq V \times V$ is the interaction set. The precise nature of these
influences are given by the \emph{local
  transition functions} $f_i: \{0,1\}^n \to \{0,1\}$ which are associated to
each automaton and satisfy:
\begin{equation}
\label{eq_effectivite}
	\forall i, j \in V,\ (j,i) \in A \iff \exists x \in \{0,1\}^n,\ f_i(x) \neq 
	f_i(\overline{x}^j)\text{.}
\end{equation}
In other words, $(j, i)$ is an arc of $G$ if and only if $j$ effectively
influences $i$. This means that in some network configurations (but not
necessarily in all of them), the state of $j$ causes a change of states of
$i$. As a consequence, a Boolean automata network is entirely defined by the set
of local transition functions of its automata.  Figure~\ref{fig_struct}
illustrates a Boolean automata network of size $2$ by picturing the set of its
local transition functions and its underlying interaction graph.
\begin{figure}[t!]
	\centerline{\begin{tabular}{ccm{30mm}}
			$\left\{
				\begin{array}{l}
					f_0(x) = x_1\\
					f_1(x) = (\neg x_0 \land \neg x_1) \lor (x_0 \land x_1)\\
				\end{array}
			\right.$ &
			~~~~~~ & 
			\begin{tikzpicture}[shorten >=0.5pt,node distance=1.5cm,on grid]
				\node[circle, draw] (n_0) {$0$};
				\node[circle, draw] (n_1) [right=of n_0] {$1$};
				\path[->]
				(n_0) edge [bend left] node {} (n_1)
				(n_1) edge [bend left] node {} (n_0)
				edge [loop right] node {} ();
 			\end{tikzpicture}
		\end{tabular}}
 		\caption{A Boolean automata network of size $2$ with its set of local transition 
 			functions and its underlying interaction graph.}
 		\label{fig_struct}
\end{figure}
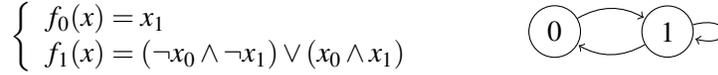

\subsection{Updating modes and transition graphs}
\label{sec_updating_modes}

The definition of a network does not determine its possible dynamical
behaviours. To do so, the way automata are updated over time has to be
specified. Here, we introduce the three distinct updating modes that are used in
this paper.\medskip

The most general standpoint is to consider every
possibility. Considering a network $N$ as a state transition system, each
configuration of $N$ is subjected to $2^{n-1}$ outgoing transitions, one for
each non-empty set of automata whose states can be updated. For any subset $ W
\neq \emptyset \subseteq V$, we define the update function $F_W: \{0,1\}^n \to
\{0,1\}^n$ such that:
\begin{equation*}
	\forall x \in \{0,1\}^n, \forall i \in V,\ F_{W}(x)_i = \begin{cases}
	f_i(x) & \text{if } i \in W \text{,}\\ 
	x_i & \text{otherwise.}
  \end{cases}
\end{equation*}
Thus, if we consider the most \emph{general updating mode}, the global network
behaviour is represented by the \emph{general transition graph} $\mathcal{G}_g =
(\{0,1\}^n, T_g)$ where $T_g = \{(x, F_W(x))\ |\ x \in \{0,1\}^n,\ W \neq
\emptyset\, \subseteq V\}$~\cite{Noual2010}. In $\mathcal{G}_g$, an arc is 
labelled by the list of subsets $W$ of automata that are updated in their 
corresponding transition $(x, F_W(x))$ such that each $F_W$ applied to $x$ gives
the same image (see the left panel of Figure~\ref{fig_updating}).\medskip
\begin{figure}[t!]
	\centerline{\begin{tabular}{c|c|c}
			\begin{tikzpicture}[shorten >=0.5pt,node distance=2.25cm,on grid]
				\node[rectangle, draw] (c0) {$\footnotesize 00$};
				\node[rectangle, draw] (c1) [right=of c0] {$\footnotesize 01$};
				\node[rectangle, draw] (c2) [below=of c0] {$\footnotesize 10$};
				\node[rectangle, draw, fill=black!20] (c3) [right=of c2] {$\footnotesize 11$};
				\tikzset{LabelStyle/.style = {fill=white,sloped}}
				\tikzset{LabelStyle/.style = {color=white, text=black}}
				\path[->]
				(c0) edge [bend left] node {\colorbox{white}{$\scriptstyle 1, \{0,1\}$}} (c1) 
				edge [loop above] node {\colorbox{white}{$\scriptstyle 0$}} ()
				(c1) edge node {\colorbox{white}{$\scriptstyle 1$}} (c0)
				edge node {\colorbox{white}{$\scriptstyle 0$}} (c3)
				edge node {\colorbox{white}{$\scriptstyle \{0,1\}$}} (c2)
				(c2) edge node {\colorbox{white}{$\scriptstyle 0, \{0,1\}$}} (c0) 
				edge [loop below] node {\colorbox{white}{$\scriptstyle 1$}} ()
				(c3) edge [loop below] node {\colorbox{white}{$\scriptstyle 0, 1, \{0,1\}$}} ();
 			\end{tikzpicture}~ &
			~~~\begin{tikzpicture}[shorten >=0.5pt,node distance=2.25cm,on grid]
				\node[rectangle, draw] (c0) {$\footnotesize 00$};
				\node[rectangle, draw] (c1) [right=of c0] {$\footnotesize 01$};
				\node[rectangle, draw] (c2) [below=of c0] {$\footnotesize 10$};
				\node[rectangle, draw, fill=black!20] (c3) [right=of c2] {$\footnotesize 11$};
				\tikzset{LabelStyle/.style = {fill=white,sloped}}
				\tikzset{LabelStyle/.style = {color=white, text=black}}
				\path[->]
				(c0) edge [bend left] node {\colorbox{white}{$\scriptstyle 1$}} (c1) 
				edge [loop above] node {\colorbox{white}{$\scriptstyle 0$}} ()
				(c1) edge node {\colorbox{white}{$\scriptstyle 1$}} (c0)
				edge node {\colorbox{white}{$\scriptstyle 0$}} (c3)
				(c2) edge node {\colorbox{white}{$\scriptstyle 0$}} (c0) 
				edge [loop below] node {\colorbox{white}{$\scriptstyle 1$}} ()
				(c3) edge [loop below] node {\colorbox{white}{$\scriptstyle 0, 1$}} ();
 			\end{tikzpicture}~~ &
			~~~\begin{tikzpicture}[shorten >=0.5pt,node distance=2.25cm,on grid]
				\node[rectangle, draw, color=white, fill=black!65] (c0) {$\footnotesize 00$};
				\node[rectangle, draw, color=white, fill=black!65] (c1) [right=of c0] {$\footnotesize 01$};
				\node[rectangle, draw, color=white, fill=black!65] (c2) [below=of c0] {$\footnotesize 10$};
				\node[rectangle, draw, fill=black!20] (c3) [right=of c2] {$\footnotesize 11$};
				\tikzset{LabelStyle/.style = {fill=white,sloped}}
				\tikzset{LabelStyle/.style = {color=white, text=black}}
				\path[->]
				(c0) edge node {\colorbox{white}{$\scriptstyle \{0,1\}$}} (c1) 
				edge [color=white, loop above] node {\colorbox{white}{\textcolor{white}{$\scriptstyle 0$}}} ()
				(c1) edge node {\colorbox{white}{$\scriptstyle \{0,1\}$}} (c2)
				(c2) edge node {\colorbox{white}{$\scriptstyle \{0,1\}$}} (c0) 
				edge [color=white, loop below] node {\colorbox{white}{\textcolor{white}{$\scriptstyle 1$}}} ()
				(c3) edge [loop below] node {\colorbox{white}{$\scriptstyle \{0,1\}$}} ();
 			\end{tikzpicture}
		\end{tabular}} 
 		\caption{(left) General, (centre) asynchronous and (right) parallel transition 
 			graphs of the Boolean automata network of Figure~\ref{fig_struct}.}
  		\label{fig_updating}
\end{figure}
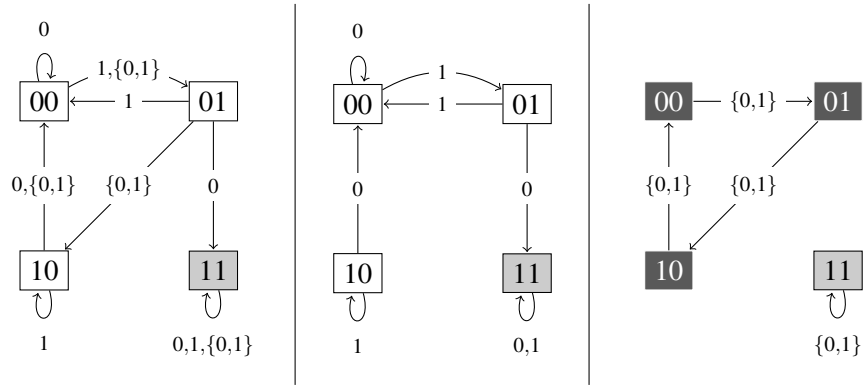

Transitions $(x, F_{i}(x))$ that involve the update of one automaton $i \in
V$ only are called \emph{asynchronous transitions}. Transitions $(x, F_W(x))$, 
$|W| > 1$, that involve the update of several are called \emph{synchronous transitions}.
The sub-graph $\mathcal{G}_a = (\{0,1\}^n, T_a)$ of $\mathcal{G}_g$ whose set of
arcs $T_a = \{(x, F_{\{i\}}(x))\ |\ x \in \{0,1\}^n,\ i \in V\}$ equals the set
of asynchronous transitions of the network is called the \emph{asynchronous 
transition graph}. Taking $\mathcal{G}_a$ as a reference transition graph allows to define
the \emph{asynchronous updating mode} according to which, in each configuration, 
only $n$ transitions are considered, one for each automaton that can be updated 
alone. This updating mode has been widely used in studies of Thomas and his 
co-workers in~\cite{Remy2008,Richard2007,Richard2004,Thomas1981}. An 
illustration of an asynchronous transition graph is given in 
Figure~\ref{fig_updating}~(centre).\medskip

The general and the asynchronous transition graphs are very large graphs. 
In some cases, to draw intuitions, it is interesting to restrict our attention to 
the transitions resulting from a specific deterministic updating schedule $u$. This 
distinct point of view is derived from the work of 
Robert~\cite{Robert1986,Robert1995} and has been adopted in various studies, 
see for instance~\cite{Aracena2009,Demongeot2010,Goles2010}. 
Section~\ref{sec_xor} focuses on a specific such deterministic mode, called the
\emph{parallel updating mode} $\pi$. It consists in updating all automata at once 
in each network configuration. The underlying global transition function is 
$F[\pi] = F_V$ so that $\forall i\in V,\ F[\pi](x)_i = f_i(x)$ and the network 
behaviour is considered to be described by the graph of $F[\pi]$, \emph{i.e.}, 
the transition graph $\mathcal{G}_{\pi} = (\{0,1\}^n, T_{\pi})$ where 
$T_{\pi} = \{(x, F[\pi](x))\ |\ x \in \{0,1\}^n\}$ (see the right panel of 
Figure~\ref{fig_updating}).

\subsection{Dynamical behaviours and non-monotony}
\label{sec_non-monotony}

Consider a Boolean automata network $N$ and an updating mode $u \in
\{g, a, \pi\}$ among those mentioned above. Let $\mathcal{G}_u$ be the
associated transition graph and $x \in \{0,1\}^n$ be a configuration of $N$. The
definitions that follow extend directly and naturally to more general updating
modes.\medskip

A path in $\mathcal{G}_u$ that starts in $x$ is a \emph{trajectory} of
$x$. In $\mathcal{G}_u$, strongly connected components that admit no outgoing
arcs, called terminal strongly connected components, are called the
\emph{attractors} of $N$. They correspond to the asymptotic behaviours of $N$.
Their sizes equal the number of configurations that they contain. The
configurations belonging to attractors of $N$ are called its \emph{recurrent
  configurations}. An attractor of size $1$ is called a \emph{stable
  configuration}. Other attractors are called \emph{stable oscillations}. In the
deterministic context of the parallel update schedule, stable configurations
correspond to fixed points of the global transition function $F[\pi]$ and stable
oscillations of size $p$ are rather called \emph{limit cycles} of \emph{period}
$p$. They correspond to oriented cycles in $\mathcal{G}_\pi$. These notions are
illustrated in Figure~\ref{fig_updating}. There, stable configurations are
represented in light grey and stable oscillations in dark grey. In particular,
this figure shows that the network defined in Figure~\ref{fig_struct} admits a
unique stable configuration, that is, configuration $11$, whatever the updating
mode chosen. This recalls that stable configurations are preserved unlike
sustained oscillations (see the limit cycle of period $3$ appearing when the
network is subjected to the parallel updating mode).\medskip

By analogy with continuous functions, the local transition function $f_i$ 
of an automaton $i \in V$ is said to be \emph{locally monotone} in $j \in V$ if, either:
\begin{equation*}
	\forall x = (x_0, \ldots, x_{n-1}) \in \{0,1\}^n,\ f_i(x_0, \ldots, x_{j-1}, 0, 
	x_{j+1}, \ldots, x_{n-1}) \leq f_i(x_0, \ldots, x_{j-1}, 1, x_{j+1}, \ldots, 
	x_{n-1})
\end{equation*}
or:
\begin{equation*}
	\forall x = (x_0, \ldots, x_{n-1}) \in \{0,1\}^n,\  f_i(x_0, \ldots, x_{j-1}, 0, 
	x_{j+1}, \ldots, x_{n-1}) \geq f_i(x_0, \ldots, x_{j-1}, 1, x_{j+1}, \ldots, 
	x_{n-1})\text{.}
\end{equation*}
In other terms, $f_i$ is locally monotone in $j$ if, in the conjunctive (or
disjunctive) normal form of $f_i(x)$, either only $x_j$ appears or only $\neg
x_j$ does.  The function $f_i$ is said to be locally monotone or simply monotone
if it is locally monotone in \emph{all} $j\in V$.  It is said to be \emph{non}
(\emph{locally}) \emph{monotone} otherwise. In this latter case, there is a $j
\in V$ such that in some configurations, the state of $i$ imitates that of $j$
and in some other configurations, on the contrary, the state of $i$ negates that
of $j$. A network is \emph{monotone} when all functions $f_i$, $i \in V$, are
monotone. Otherwise, if at least one local transition function is non-monotone,
the network is said to be \emph{non-monotone}. Note that we distinguish
\emph{totally non-monotone} networks (with only non-monotone local transition
functions) from \emph{partially non-monotone} networks (composed by at least one
local monotone transition function). As an example, the network of
Figure~\ref{fig_struct} is partially non-monotone.

\section{Synchronism sensitivity and non-monotony}
\label{sec_synchro}

The aim of this section is to focus on the concept of synchronism 
sensitivity of Boolean automata networks and highlight that non-monotony is a 
consistent structural parameter that has a significant role in this line.

\subsection{Synchronism sensitivity cases}
\label{sec_synchro_cases}

When Boolean automata networks are viewed in the framework of state
transition systems by means of the general and asynchronous updating modes,
questions about the influence of synchronism on the dynamical behaviours of
networks naturally arise. The notion of \emph{synchronism sensitivity} of a
network can then be informally described as the fact that its dynamical
behaviour  changes significantly when synchronism is taken into
account in the computation of its evolution. On the basis of the transition
graphs $\mathcal{G}_a$ and $\mathcal{G}_g$ (and more precisely on what can
change by building $\mathcal{G}_g$ from $\mathcal{G}_a$), we concentrate on
asymptotic dynamical behaviours (and specifically on recurrent configurations
rather than attractors). First, we describe the different cases that can
possibly occur when synchronous transitions are added to an asynchronous
transition graph.\medskip

Let $N$ be a Boolean automata network with its associated asynchronous
and general transition graphs $\mathcal{G}_a = (\{0,1\}^n, T_a)$ and
$\mathcal{G}_g = (\{0,1\}^n, T_g)$ and let $x, y \in \{0,1\}^n$ be two distinct
configurations of $N$.  We say that a synchronous transition from $(x, y) \in
T_g$ is \emph{sequentialisable} if there exists a sequence of asynchronous
transitions from $x$ to $y$, \emph{i.e.}, if there is a trajectory from $x$ to
$y$ in $\mathcal{G}_a$. It is obvious that if all synchronous transitions of
$\mathcal{G}_g$ are sequentialisable then, adding synchronism does not change
the asymptotic dynamical behaviour of $N$ and $N$ is then not synchronism
sensitive. Let us therefore restrict the study to the case where $\mathcal{G}_g$
contains a non-sequentialisable synchronous transition $(x, y)$. For any $z \in
\{0,1\}^n$, we let $\mathcal{A}_z$ (resp.  $\mathcal{A}_z^\star$) be the set of
attractors to which z leads or belongs in $\mathcal{G}_a$ (resp. in
$\mathcal{G}_g$). And we denote by $\mathcal{L} = \cup_{z \in \{0,1\}^n}
\mathcal{A}_z$ (resp. $\mathcal{L}^\star = \cup_{z \in \{0,1\}^n}
\mathcal{A}_z^\star$) the set of all attractors in $\mathcal{G}_a$ (resp. in
$\mathcal{G}_g$). With these notations, because of the existence of transition
$(x, y)$ in $\mathcal{G}_g$, any attractor that can be reached by $y$ can also
be by $x$ so $\mathcal{A}_y^\star \subseteq \mathcal{A}_x^\star$. On the
contrary, in $\mathcal{G}_a$, because there are no trajectories from $x$ to $y$
($(x, y)$ is non-sequentialisable), $\mathcal{A}_y \subsetneq \mathcal{A}_x$ is
impossible. Indeed, either $(i)$ $y$ is transient and the only attractors that
it can reach are those of $\mathcal{A}_y = \mathcal{A}_x$ that can be reached
from $x$, either $(ii)$ $y$ is transient and it can reach attractors in
$\mathcal{A}_y \setminus \mathcal{A}_x \neq \emptyset$ that cannot be reached
from $x$, or $(iii)$ $y$ is recurrent and since there are no trajectories from
$x$ to $y$, there also are no trajectories from $y$ to $x$ also is and $\mathcal{A}_x
\cap \mathcal{A}_y = \emptyset$. Notice that, in the two latter cases induce
$\mathcal{A}_y \nsubseteq \mathcal{A}_x$ and that $(i)$, $(ii)$ and $(iii)$
respectively yield cases $2$, $3$ and $4$ listed below. Thus, when the
non-sequentialisable synchronous transition $(x, y)$ is added to
$\mathcal{G}_a$, one of the only four possible cases listed below holds:
\begin{enumerate}
\item $x$ is transient in $\mathcal{G}_a$. Consequently, the set $\mathcal{L} = 
	\mathcal{L}^\star$ of all attractors is unchanged. All configurations $z \in 
	\{0,1\}^n$ that can reach $x$ in $\mathcal{G}_a$, including $x$, remain 
	transient but gain the possibility to reach attractors in $\mathcal{A}_y 
	\setminus \mathcal{A}_z$ (\emph{i.e.}, $\mathcal{A}_y = \mathcal{A}_y^\star$ 
	and $\mathcal{A}_x \subseteq \mathcal{A}_z \implies \mathcal{A}_z^\star = 
	\mathcal{A}_z \cup \mathcal{A}_y$).
\item $x$ is recurrent, $y$ is transient and $\mathcal{A}_y = \mathcal{A}_x$. 
	Consequently, all $z \in \{0,1\}^n$ on a trajectory from $y$ to $\mathcal{A}_x$, 
	including $y$, become recurrent and are included in $\mathcal{A}_x^\star$, 
	causing $\mathcal{A}_x$ to grow (to become $\mathcal{A}_x^\star$).
\item $x$ is recurrent, $y$ is transient and $\mathcal{A}_y \setminus \mathcal{A}_x 
	\neq \emptyset$. Then $x$ becomes transient causing $\mathcal{L}$ to loose 
	attractor $\mathcal{A}_x$ ($\mathcal{A}_x^\star = \mathcal{A}_y = 
	\mathcal{A}_y^\star$ and $\mathcal{L}^\star = \mathcal{L} \setminus \mathcal{A}_x$).
\item both $x$ and $y$ are recurrent in $\mathcal{G}_a$. Attractor $\mathcal{A}_x$ 
	``empties itself'' in $\mathcal{A}_y$ ($\forall z \in \mathcal{A}_x$, $z$ becomes 
	transient and such that $\mathcal{A}_z = \mathcal{A}_x \nsubseteq 
	\mathcal{A}_z^\star = \mathcal{A}_y^\star$) also causing $\mathcal{L}$ to loose attractor 
	$\mathcal{A}_x$ (to become $\mathcal{L}^\star$).
\end{enumerate}

\subsection{Synchronism sensitivity levels}
\label{sec_synchro_levels}

The four cases above suggest between three and four levels of sensitivity 
(see Definition~\ref{def_level} below)  that a Boolean automata network can have to the 
addition of synchronism (the relative importance of levels $1^\circ$ and $1^\bullet$ 
being disputable, they are deliberately not ordered). Cases $1$ and $2$ respectively 
yield levels $1^\circ$ and $1^\bullet$ and cases $3$ and $4$ both yield level $2$.
\begin{definition}
	\label{def_level}
	Let $N$ be a Boolean automata network. The \emph{synchronism sensitivity} of 
	$N$ can be of:
	\begin{itemize}
	\item \emph{level $0$}: $N$ is not sensitive at all. All its synchronous 
		transitions either act as shortcuts for asynchronous trajectories or, on the 
		contrary, add local, confluent deviations which increase the number of possible 
		steps in a trajectory without changing its outcome.
	\item \emph{level $1^\circ$}: $N$ is sensitive in the sense that the addition of 
		synchronism grants additional liberty in the evolutions of some transient 
		configurations that are made to reach a greater number of different attractors.
	\item \emph{level $1^\bullet$}: $N$ is sensitive in the sense that the addition of 
	synchronism causes some transient configurations to become recurrent and thus some 
	(necessarily unstable) attractors to grow.
	\item \emph{level $2$}: $N$ is sensitive in the sense that the addition of synchronism 
	destroys attractors.
	\end{itemize}
\end{definition}

As said before, because we focus here exclusively on recurrent
configurations, the only networks that we have to consider are of levels
$1^\bullet$ and $2$. However, our recent studies have shown that, contrary to
level $2$, level $1^\bullet$ comprises many networks. It is thus not
sufficiently discriminant and consequently not significant in our framework. So,
let us concentrate on  level $2$.

\subsection{Synchronism sensitive minimal networks}
\label{sec_synchro_result}

Focusing on synchronism sensitive Boolean automata networks of level $2$, 
our aim is now to  show what are the minimal networks (in terms of size) which are
sensitive to the addition of synchronism and how they  relate to non-monotony. 
Here, the motivation directly comes from systems and synthetic biology where the 
discovering of minimal genetic interaction patterns with singular dynamical properties 
(\emph{i.e.}, singular biological functionalities) seems central to improve our 
understanding of living organisms.\medskip

This leads us to the following proposition.
\begin{proposition}
	\label{prop_asynchronous-general} 
	The minimal Boolean automata networks that are synchronism sensitive of level $2$ 
	are totally non-monotone.
\end{proposition}

\begin{figure}[t!]
	\centerline{\hspace*{-6mm}\begin{tabular}{ccm{35.1mm}}
			$f_0, f_1 \in \{x \mapsto (x_0 \oplus x_1), x \mapsto \neg (x_0 \oplus x_1)\}$ &
			~~~~~~ & 
			\begin{tikzpicture}[shorten >=0.5pt,node distance=1.5cm,on grid]
				\node[circle, draw] (n_0) {$0$};
				\node[circle, draw] (n_1) [right=of n_0] {$1$};
				\path[->]
				(n_0) edge [bend left] node {} (n_1)
				edge [loop left] node {} ()
				(n_1) edge [bend left] node {} (n_0)
				edge [loop right] node {} ();
 			\end{tikzpicture}
		\end{tabular}}\medskip
 		\hrule\bigskip
 		
		\centerline{\begin{tabular}{c|c}
			\begin{tikzpicture}[shorten >=0.5pt,node distance=1.75cm,on grid]
				\node[rectangle, draw, fill=black!20] (c0) {$\footnotesize \overline{x}^{0,1}$};
				\node[rectangle, draw, color=white, fill=black!65] (c1) [right=of c0] {$\footnotesize \overline{x}^{0}$};
				\node[rectangle, draw, color=white, fill=black!65] (c2) [below=of c0] {$\footnotesize \overline{x}^{1}$};
				\node[rectangle, draw, color=white, fill=black!65] (c3) [right=of c2] {$\footnotesize x$};
				\tikzset{LabelStyle/.style = {fill=white,sloped}}
				\tikzset{LabelStyle/.style = {color=white, text=black}}
				\path[->]
				(c0) edge [loop left] node {\colorbox{white}{$\scriptstyle 0,1$}} (c1)
				(c1) edge node {\colorbox{white}{$\scriptstyle 0$}} (c3)
				edge [loop right] node {\colorbox{white}{$\scriptstyle 1$}} ()
				(c2) edge node {\colorbox{white}{$\scriptstyle 1$}} (c3)
				edge [loop left] node {\colorbox{white}{$\scriptstyle 0$}} ()
				(c3) edge [bend right] node {\colorbox{white}{$\scriptstyle 0$}} (c1)
				edge [bend left] node {\colorbox{white}{$\scriptstyle 1$}} (c2);
 			\end{tikzpicture}~~~~~~~ &
			~~~~~~~\begin{tikzpicture}[shorten >=0.5pt,node distance=1.75cm,on grid]
				\node[rectangle, draw, fill=black!20] (c0) {$\footnotesize \overline{x}^{0,1}$};
				\node[rectangle, draw] (c1) [right=of c0] {$\footnotesize \overline{x}^{0}$};
				\node[rectangle, draw] (c2) [below=of c0] {$\footnotesize \overline{x}^{1}$};
				\node[rectangle, draw] (c3) [right=of c2] {$\footnotesize x$};
				\tikzset{LabelStyle/.style = {fill=white,sloped}}
				\tikzset{LabelStyle/.style = {color=white, text=black}}
				\path[->]
				(c0) edge [loop left] node {\colorbox{white}{$\scriptstyle 0,1$}} (c1)
				(c1) edge node {\colorbox{white}{$\scriptstyle 0$}} (c3)
				edge [loop right] node {\colorbox{white}{$\scriptstyle 1$}} ()
				(c2) edge node {\colorbox{white}{$\scriptstyle 1$}} (c3)
				edge [loop left] node {\colorbox{white}{$\scriptstyle 0$}} ()
				(c3) edge [bend right] node {\colorbox{white}{$\scriptstyle 0$}} (c1)
				edge [bend left] node {\colorbox{white}{$\scriptstyle 1$}} (c2)
				edge node {\colorbox{white}{$\scriptstyle \{0, 1\}$}} (c0);
 			\end{tikzpicture}  			
		\end{tabular}}
	\caption{(top) Generic description of the four smallest Boolean automata networks 
		that satisfy the conditions of Proposition~\ref{prop_asynchronous-general}. 
		(bottom) Generic (left) asynchronous and (right) general transition graphs of 
		these networks.}
	\label{fig_motiv}
\end{figure}

\begin{proof} 
	As explained above, in order to be synchronism sensitive of level
	$2$, a Boolean automata network needs to have at least one non-sequentialisable
	synchronous transition in its general transition graph $\mathcal{G}_g$. Let us
	uncover the structural conditions that must be satisfied by a minimal network
	$N$ belonging to level $2$, with at least one non-sequentialisable synchronous
	transition. First, $N$ needs to have more than one automaton because, if not,
	synchronism has no sense. If it has size $2$, then, to have a
	non-sequentialisable synchronous transition, $\mathcal{G}_g$ needs to contain a
	generic sub-graph (with asynchronous transitions only) of the following
	form:\\[1.5mm] \centerline{
		\begin{tikzpicture}[shorten >=0.5pt,node distance=1.75cm,on grid]
			\node[rectangle, draw] (c0) {$\footnotesize \overline{x}^{0,1}$};
			\node[rectangle, draw] (c1) [right=of c0] {$\footnotesize \overline{x}^{0}$};
			\node[rectangle, draw] (c2) [below=of c0] {$\footnotesize \overline{x}^{1}$};
			\node[rectangle, draw] (c3) [right=of c2] {$\footnotesize x$};
			\tikzset{LabelStyle/.style = {fill=white,sloped}}
			\tikzset{LabelStyle/.style = {color=white, text=black}}
			\path[->]
			(c1) edge [loop right] node {\colorbox{white}{$\scriptstyle \{1\}$}} ()
			(c2) edge [loop left] node {\colorbox{white}{$\scriptstyle \{0\}$}} ()
			(c3) edge node {\colorbox{white}{$\scriptstyle \{0\}$}} (c1)
			edge node {\colorbox{white}{$\scriptstyle \{1\}$}} (c2);
 		\end{tikzpicture}
	}\\[1.5mm] where $\overline{x}^{\, i,j} = \overline{x}^{\, \{i,j\}} =
	\overline{\overline{x}^i}^j$ (see Equation~\ref{eq_notations_basiques}). This
	sub-graph is the smallest that is necessary for the general transition graph to
	contain a non-sequentialisable synchronous transition $(x,\overline{x}^{\,
	  i,j})$. It is also easy to see that there can be only one non-sequentialisable
	synchronous transition in the general transition graph $\mathcal{G}_g$ of a
	network of size and level $2$.  Moreover, to guarantee synchronism sensitivity
	of level $2$, because fixed points are conserved whatever the updating mode, the
	synchronous transition $(x, \overline{x}^{i,j})$ must go out of a set of
	configurations belonging to an asynchronous stable oscillation. Now, there is
	only one way to create an asynchronous stable oscillation that verifies the
	presence of the asynchronous sub-graph drawn above. It consists in adding
	transitions $(\overline{x}^{i},x)$ and $(\overline{x}^{j},x)$. On this basis, in
	order to create synchronism sensitivity of level $2$, configuration
	$\overline{x}^{\, i,j}$ needs to be a fixed point of $N$.  If not,
	$\overline{x}^{\, i,j}$ is a predecessor of the limit cycle and adding
	synchronism will maintain the recurrence of every asynchronous recurrent
	configuration. Thus, since $\overline{x}^{\, i,j}$ is a fixed point of $N$,
	adding transition $(x,\overline{x}^{i,j})$ makes $\overline{x}^{ i,j}$ become
	the only attractor of $N$ with respect to the general updating mode. Thus, the
	general transition graph of $N$ must have the form  pictured in the
	bottom right panel of Figure~\ref{fig_motiv} (the bottom left panel illustrates
	the asynchronous transition graph of such a $N$ to compare). Hence, only two
	functions $f_0$ are possible. If in configuration $x$ above, $x_0 = 1$, then,
	$f_0(x): x \mapsto x_0 \oplus x_1$ where $\oplus$ denotes the \textsc{xor}
	connector\footnote{$\forall a,b\in \{0,1\},\, a\oplus b= (a\land \neg b) \vee
	  (\neg a \land b)$.}. If in configuration $x$ above, $x_0 = 0$, then $f_0(x): x
	\mapsto \neg ( x_0 \oplus x_1)$. The function $f_1$ is defined similarly. In
	conclusion, there are four minimal networks satisfying the properties of
	Proposition~\ref{prop_asynchronous-general}. They have size $2$ and their
	interaction graphs equal the graph pictured in the top panel of
	Figure~\ref{fig_motiv}. Their two local interaction functions $f_0$ and $f_1$
	either equal $x \mapsto x_0 \oplus x_1$ or $x \mapsto \neg ( x_0 \oplus
	x_1)$.
\end{proof}

Among the four minimal Boolean automata networks described in the proof 
above that are synchronism sensitive of level $2$, those defined by
$$\left\{
	\begin{array}{l}
		f_0(x) = x_0 \oplus x_1\\
		f_1(x) = \neg (x_0 \oplus x_1)\\
	\end{array}
\right. \qquad \text{and} \qquad \left\{ 
	\begin{array}{l}
		f_0(x) = \neg (x_0 \oplus x_1)\\
		f_1(x) = x_0 \oplus x_1\\
	\end{array}
\right.$$
are isomorphic. This result relates intimately synchronism sensitivity to 
non-monotony. Indeed, the smallest patterns that produce synchronism sensitivity
strong singularities are non-monotone networks. Moreover, it is easy to see that 
synchronism sensitivity of level $2$ applies to other non-monotone networks. 
Thus, it would be judicious and interesting to go further and characterise the 
family of synchronism sensitive non-monotone networks of level $2$. Now, in order 
to develop intuition about the dynamical behaviour of general non-monotone 
networks, we choose to focus on a specific class of non-monotone networks, namely 
\textsc{xor} circulant networks.

\section{{\normalsize XOR} circulant networks}
\label{sec_xor}

Let us focus now on the trajectorial and asymptotic dynamical behaviours 
of \textsc{xor} circulant networks. These networks define a class of non-monotone 
Boolean automata networks that is not too large but  has all the necessary properties to 
present complex behaviours.

\subsection{Definitions and basic properties}
\label{sec_xor_def}

A matrix $\mathcal{C}$ of order $n$ whose $i^{\text{th}}$ row vector 
$\mathcal{C}_i$ ($i<n$) is the right-cyclic permutation with offset $i$ of its first 
row vector $\mathcal{C}_0$ so that $\mathcal{C}$ has the following form:
\begin{equation*}
	\mathcal{C} = \begin{pmatrix}
		c_0 & c_1 & c_2 & \ldots & c_{n-1}\\
		c_{n-1} & c_0 & c_1 & \ldots & c_{n-2}\\
		c_{n-2} & c_{n-1} & c_0 & \ldots & c_{n-3}\\
		\vdots & \vdots & \vdots & \ddots & \vdots \\
		c_1 & c_2 & c_3 & \ldots & c_{0}\\
	\end{pmatrix}\text{}
\end{equation*}
is a \emph{circulant matrix}. For any integer $k \geq 2$, a \emph{$k$-\textsc{xor} 
circulant network} of size $n \geq k$ is a network with $n$ automata so that the 
following four properties are satisfied: 
\begin{enumerate}
\item[\textit{(1)}] the adjacency matrix $\mathcal{C}$ of the 
	network interaction graph $G = (V, A)$, called the \emph{interaction matrix}, is a 
	circulant matrix;
\item[\textit{(2)}] each row $\mathcal{C}_{i}$ of this matrix 
	contains exactly $k$ non-null coefficients;
\item[\textit{(3)}] $\mathcal{C}_{0,n-1} = c_{n-1} = 1$;
\item[\textit{(4)}] the local transition function of any 
	automaton $i$ is a \textsc{xor} function such that $\forall x \in \{0,1\}^n,\ 
	f_i(x) = \bigoplus_{j\in V} \mathcal{C}_{i,j} \cdot x_j = \sum_{j\in V} 
	\mathcal{C}_{i,j} \cdot x_j \mod 2$.
\end{enumerate}
Here, \textsc{xor} circulant networks are subjected to the \emph{parallel updating 
mode}, which means that if $x = x(t)$ is the configuration at time step $t$, then 
the network configuration at time step $t + 1$ equals $x(t+1) = F(x) = \mathcal{C} 
\cdot x$ (where operations are taken modulo $2$). Notice that $F$ is then a linear 
function~\cite{Cull1971,Elspas1959,Toledo2005} and that, consequently, a \textsc{xor} 
circulant network is entirely defined by its interaction graph $G = (V, A)$ or by its 
interaction matrix $\mathcal{C}$. Figure~\ref{fig_xor-network} pictures two 
interaction graphs, the first one (left panel) is a $2$-\textsc{xor} circulant 
network of size $4$, the second is a $3$-\textsc{xor} circulant network of size 
$6$.\medskip
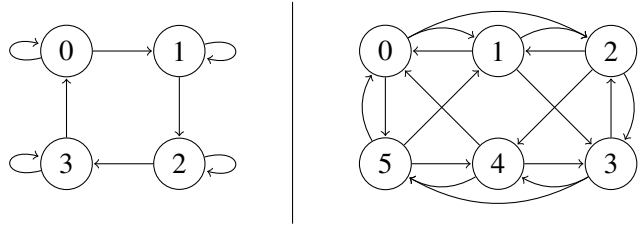
\begin{figure}[t!]
	\centerline{\begin{tabular}{m{38mm}|m{44.5mm}}
			\begin{tikzpicture}[shorten >=0.5pt,node distance=1.5cm,on grid]
				\node[circle, draw] (n_0) {$0$};
				\node[circle, draw] (n_1) [right=of n_0] {$1$};
				\node[circle, draw] (n_2) [below=of n_1] {$2$};
				\node[circle, draw] (n_3) [below=of n_0] {$3$};
				\path[->]
				(n_0) edge node {} (n_1)
				edge [loop left] node {} ()
				(n_1) edge node {} (n_2)
				edge [loop right] node {} ()
				(n_2) edge node {} (n_3)
				edge [loop right] node {} ()
				(n_3) edge node {} (n_0)
				edge [loop left] node {} ();
 			\end{tikzpicture} &
 			~~~~~~~\begin{tikzpicture}[shorten >=0.5pt,node distance=1.5cm,on grid]
				\node[circle, draw] (n_0) {$0$};
				\node[circle, draw] (n_1) [right=of n_0] {$1$};
				\node[circle, draw] (n_2) [right=of n_1] {$2$};
				\node[circle, draw] (n_3) [below=of n_2] {$3$};
				\node[circle, draw] (n_4) [below=of n_1] {$4$};
				\node[circle, draw] (n_5) [below=of n_0] {$5$};
				\path[->]
				(n_0) edge [bend left] node {} (n_1)
				edge [bend left] node {} (n_2)
				edge node {} (n_5)
				(n_1) edge [bend left] node {} (n_2)
				edge node {} (n_3)
				edge node {} (n_0)
				(n_2) edge [bend left] node {} (n_3)
				edge node {} (n_4)
				edge node {} (n_1)
				(n_3) edge [bend left] node {} (n_4)
				edge [bend left] node {} (n_5)
				edge node {} (n_2)
				(n_4) edge [bend left] node {} (n_5)
				edge node {} (n_0)
				edge node {} (n_3)
				(n_5) edge [bend left] node {} (n_0)
				edge node {} (n_1)
				edge node {} (n_4);
 			\end{tikzpicture}
 		\end{tabular}}
	\caption{(left) A $2$-\textsc{xor} circulant network of size $4$ and (right) a 
		$3$-\textsc{xor} circulant network of size $6$.}
	\label{fig_xor-network}
\end{figure}

A $k$-\textsc{xor} circulant network $N$ can be viewed as a \emph{cellular 
automaton}. Indeed, if $N$ has size $n$ and interaction graph $G = (V, A)$, $N$ can 
be modelled by the finite one-dimensional cellular automaton of $n$ cells assimilated 
to the $n$ automata of $N$ and that satisfies what follows. The \emph{neighbourhood} 
$\mathcal{N}$ of a cell $i \in V$ equals the in-neighbourhood of automaton $i$ in $N$: 
$\mathcal{N} = \{j \in V\ |\ (j,i) \in A\}$. The local rule $\gamma: 
\{0, 1\}^{|\mathcal{N}|} \to \{0,1\}$ of the cellular automaton is similar to the 
local transition functions of $N$ and is defined as $\gamma((x_\ell)_{\ell \in 
\mathcal{N}}) = \bigoplus_{\ell \in \mathcal{N}} x_{\ell}$. We make specific use of 
this formalisation to exploit tools of the theory of cellular automata. Thus, if
$x = x(0) \in \{0,1\}^n$ is an initial configuration of $N$, we consider the 
corresponding \emph{space-time diagram}, that is, the grid of $\{0,1\}^n \times 
\mathbb{N}$ whose line $t \in \mathbb{N}$ represents $x(t)$. The \emph{trace} of 
cell or automaton $i$ then corresponds to column $i$ of this grid, that is, to the 
sequence $(x_i(t))_{t \in \mathbb{N}}$. Furthermore, for an arbitrary configuration 
$x$ and an automaton $i$, $R_i(x)$ denotes the configuration that satisfies $\forall
j \in V,\ R_i(x)_j = x_{2i-j \mod n}$ and is called the \emph{reflection of $x$} with 
respect to $i$. We write $\widetilde{N}$ to denote the \emph{reflection of
$N$}, \emph{i.e.}, the $k$-\textsc{xor} circulant network whose interaction matrix
is $\mathstrut^t\mathcal{C}$. In the sequel, unless it is made explicit, 
$\mathcal{N}^-(i)$ (resp. $\mathcal{N}^+(i)$) denotes the \emph{in-neighbourhood} 
(resp. the \emph{out-neighbourhood}) of automaton $i$ and
$\widetilde{\mathcal{N}}^-(i)$ (resp. $\widetilde{\mathcal{N}}^+(i)$) denotes
its in-neighbourhood (resp. its out-neighbourhood) in $\widetilde{N}$. Thus,
for any two automata $i, j$, $j \in \mathcal{N}^-(i) \iff j \in
\widetilde{\mathcal{N}}^+(i)$. $\widetilde{F}$ denotes the global transition 
function of $\widetilde{N}$ if $F$ denotes that of $N$. Notice that $\widetilde{F}$ 
represents the reflected global transition function of $N$. By default, unless $N$ is 
the reflection of another $k$-\textsc{xor} circulant network that
was introduced before, its automata are supposed to be numbered as suggested above, 
\emph{i.e.}, so that $c_{n-1}=\mathcal{C}_{0,n-1}=1$. This way, $\{(i,i+1 \mod n)\ |\ 
i\in V\}\subseteq A$ defines a Hamiltonian circuit in the structure of $N$ and 
$\{(i+1 \mod n, i)\ |\ i\in V\}\subseteq A$ defines a Hamiltonian circuit in the 
structure of its reflection $\widetilde{N}$.\medskip

Let us now list in the proposition below some basic properties of 
\textsc{xor} circulant networks that follow directly from the definitions of 
\textsc{xor} functions and circular matrices.
\begin{proposition}
	\label{prop_basic}
	~
	\begin{enumerate}
	\item[1.]\label{prop_number} The number of $k$-\textsc{xor} circulant networks
		of size $n$ equals $\binom{n-1}{k-1}$.
	\end{enumerate}
	Any $k$-\textsc{xor} circulant network of size $n$ satisfies the following
	properties:
	\begin{enumerate}
	\item[2.]\label{prop_0} Configuration $(0, \ldots, 0)$ is a stable
		configuration.
	\item[3.]\label{prop_1} Configuration $(1,\ldots,1)$ is a predecessor of
		$(0,\ldots,0)$ if $k$ is even and it is a stable configuration if $k$ is odd.
	\item[4.]\label{prop_isomorphism} The trajectory of a configuration $x$ is
		isomorphic to that of any configuration $y$ which is a circular permutation
		of $x$.
	\end{enumerate}
\end{proposition}

\subsection{Results}

In what follows, unless it is mentioned, the automata are always taken 
modulo the size $n$ of the network considered.

\subsubsection{General $k$-\textsc{xor} circulant networks}

Here, we concentrate on general $k$-\textsc{xor} circulant networks and 
exploit the cellular automata formalisation presented above to derive some features 
of the dynamical behaviours of these networks.
\begin{lemma}
	\label{lem_mask}
	Let $N$ be a $k$-\textsc{xor} circulant network of size $n$ with automata set
	$V$ and reflected global transition function $\widetilde{F}$. For any
	automaton $i$, let $M_i(t)$ denote the set of	automata which have state $1$ in 
	configuration $\widetilde{F}^t(\overline{0}^i)$. Then, $\forall x(0)\in \{0,1\}^n, 
	\forall t \in \mathbb{N},\ x_i(t) = \bigoplus_{j \in M_i(t)} x_j(0)$.
\end{lemma}
\begin{proof}  
	We prove Lemma~\ref{lem_mask} by induction on $t \in \mathbb{N}$. For 
	$t = 0$, $M_i(0) = \{i\}$ holds by definition of configuration $\overline{0}^i$. Thus, 
	$\forall x(0) \in \{0,1\}^n,\ x_i(0) = \bigoplus_{j \in M_i(0)} x_j(0)$. Now, 
	suppose that $\forall x(0) \in \{0,1\}^n,\ x_i(t) = \bigoplus_{j \in M_i(t)} x_j(0)$ 
	and consider the initial configuration $y(0) \in \{0, 1\}^n$. Since $y(t+1) = 
	\widetilde{F}^{t+1}(y(0)) = \widetilde{F}^t(y(1))$, the induction hypothesis applied 
	to configuration $x(0) = y(1)$ yields $y_i(t+1) = \bigoplus_{j \in M_i(t)} y_j(1)$.
	By definition, $\forall j \in V,\ y_j(1) = f_j(y(0)) = \bigoplus_{\ell \in
	\mathcal{N}^-(j)} y_\ell(0) = \bigoplus_{\ell \in \widetilde{\mathcal{N}}^+
	(j)} y_\ell(0)$. Thus, because the \textsc{xor} connector is commutative and 
	associative, we have:
	\begin{equation*}
		y_i(t+1)\ =\ \bigoplus_{j \in M_i(t)} \big( \bigoplus_{\ell \in 
		\widetilde{\mathcal{N}}^+(j)} y_\ell(0)\big)\ =\ \bigoplus_{\{\ell \text{ s.t. } 
		|\widetilde{\mathcal{N}}^-(\ell)\, \cap\, M_i(t)|\, \equiv\, 1 \mod 2\}} 
		y_\ell(0)\text{.}
	\end{equation*}
	Now, remark that $\forall t\in\mathbb{N},\ \widetilde{F}(\overline{0}^{M_i(t)}) = 
	\overline{0}^{M_i(t+1)}$ by 	definition. Then, $\forall \ell \in V, 
	\overline{0}^{M_i(t+1)}_\ell= 1$ if and only if $|\widetilde{\mathcal{N}}^-(\ell) 
	\cap M_i(t)|\equiv 1 \mod 2$. From this, we derive that $y_i(t+1) = \bigoplus_{j 
	\in M_i(t+1)} y_j(0)$ and then 	$\forall t \in \mathbb{N},\ x_i(t) = \bigoplus_{j 
	\in M_i(t)} x_j(0)$.
\end{proof}

\begin{lemma}
	\label{lem_symmetric-network} 
	Let $N$ be a $k$-\textsc{xor} circulant network of size $n$ with automata set $V$ 
	and global transition function $F$. For any automaton $i$ and for any configuration 
	$x \in \{0,1\}^n$, $\widetilde{F}(R_i(x)) = R_i(F(x))$ holds.
\end{lemma}  
\begin{proof}
	For any automaton $j$, the following holds:
	\begin{equation*}
		\widetilde{F}(R_i(x))_j\ =\ \bigoplus_{\ell \in \widetilde{\mathcal{N}}^-(j)} 
		(R_i(x))_\ell\ =\ \bigoplus_{\ell \in \widetilde{\mathcal{N}}^- (j)} x_{2i-\ell}\ 
		=\ \bigoplus_{\{\ell \text{ s.t. } 2i-\ell\, \in\, \widetilde{\mathcal{N}}^-(j)\}}
		x_{\ell}\ =\ \bigoplus_{\{\ell \text{ s.t. } j\, \in\, \mathcal{N}^-(2i-\ell)\}} x_{\ell}\text{.}
	\end{equation*} 
	Now, if $j \in \mathcal{N}^-(2i-\ell)$, then all automata $a, a' \in V$ of $N$ such that 
	$a-a' = j-(2i-\ell)$ are such that $a \in \mathcal{N}^-(a')$. In particular, if automaton 
	$j \in \mathcal{N}^-(2i-\ell)$, then $\ell \in \mathcal{N}^-(2i-j)$. Hence, we have:
	\begin{equation*}
		\bigoplus_{\{\ell \text{ s.t. } j\, \in\, \mathcal{N}^-(2i-\ell)\}} x_{\ell}\ =\ 
		\bigoplus_{\ell\, \in\, \mathcal{N}^-(2i-j)} x_{\ell} = F(x)_{2i-j}\ =\ 
		(R_i(F(x)))_j\text{,}
	\end{equation*}
	and Lemma~\ref{lem_symmetric-network} follows.
\end{proof}

\begin{proposition}
	\label{prop_symmetric-space-time-diagram} 
	Let $N$ be a $k$-\textsc{xor} circulant network of size $n$ with automata set $V$ 
	and global transition function $F$. For any automaton $i$ and for the initial 
	configuration $x(0)=\overline{0}^i$, it holds that $\forall t \in \mathbb{N}, 
	\widetilde{F}^t(x(0)) = R_i(x(t))$.
\end{proposition}
\begin{proof}
	Proposition~\ref{prop_symmetric-space-time-diagram} is proven by 
	induction on $t$. Let $t = 0$. Property $\widetilde{F}^t(x(0)) = R_i(x(t))$ is true 
	because $x(0) = \overline{0}^i$.  Suppose that it is true for $t$. Then, we have 
	$\widetilde{F}^{t+1}(x(0)) = \widetilde{F}(\widetilde{F}^t(x(0))) = 
	\widetilde{F}(R_i(x(t))$. By Lemma~\ref{lem_symmetric-network}, 
	$\widetilde{F}(R_i(x(t)) = R_i(F(x(t)) = R_i(x(t+1))$, which is the expected 
	result.
\end{proof}
\begin{figure}[t!]
	\centerline{
			\begin{tabular}{ccc}
			\includegraphics[scale=0.8]{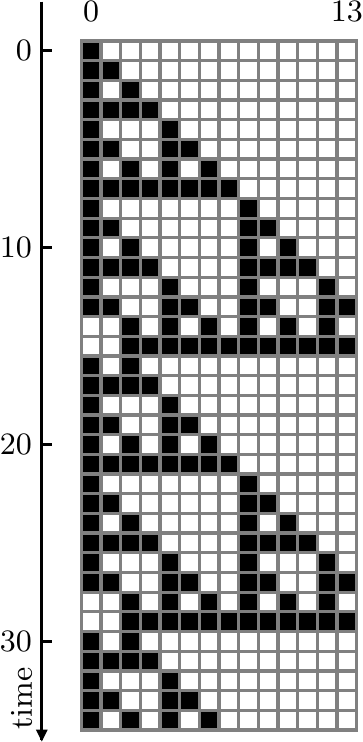}~ &
			~~\includegraphics[scale=0.8]{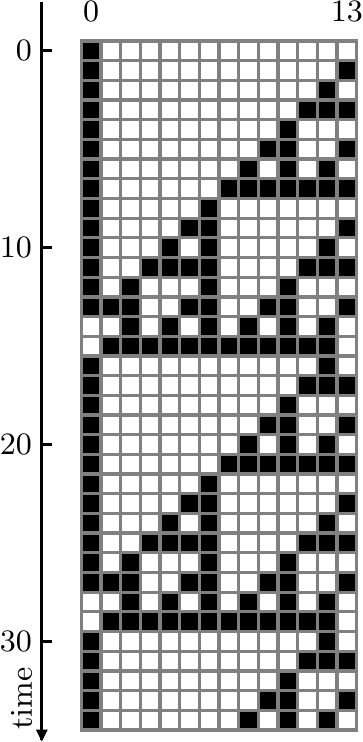}~ &
			~~\includegraphics[scale=0.8]{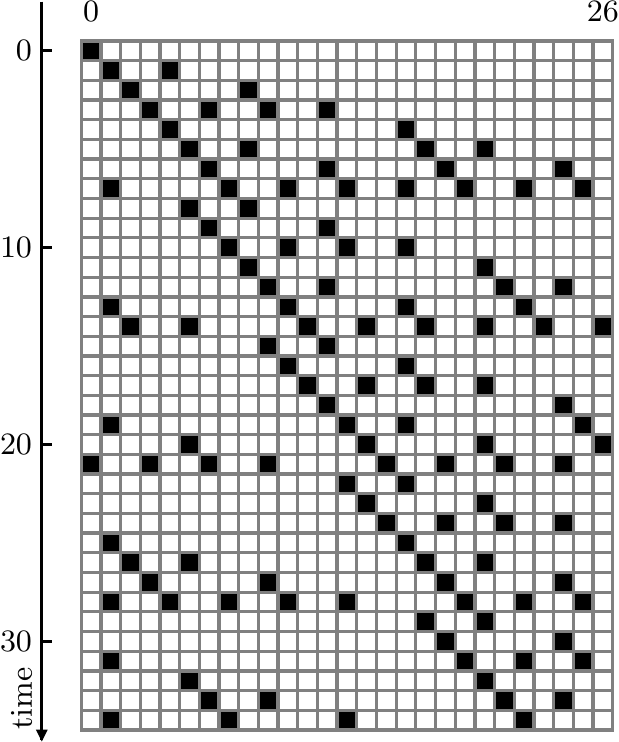}\\
			$(a)$ & $(b)$ & $(c)$
		\end{tabular}
	}
	\caption{Space-time diagrams $(a)$ of a $2$-\textsc{xor} circulant network of
	size $14$ and interaction-step $s=0$ (cf. page~\pageref{sec_twoxor}), $(b)$
	of its reflected network and $(c)$ of another $2$-\textsc{xor} circulant
	network of size $27$ and interaction-step $4$.}
	\label{fig_xor-ca-17-24}
\end{figure}

This result comes from the fact that $F$ and $\widetilde{F}$ are the global 
transition functions of two reflected $k$-\textsc{xor} circulant networks that are 
isomorphic by definition (see Figure~\ref{fig_xor-ca-17-24}). 
Proposition~\ref{prop_symmetric-space-time-diagram} implies that, for any automaton 
$i$, the space-time diagram of $(\overline{0}^{i}(t))_{t \in \mathbb{N}}$ is the reflected 
space-time diagram of $(\overline{0}^{M_i(t)})_{t \in \mathbb{N}}$ with respect to $i$ 
and is related to the trace of automaton $i$. Thus, the space-time diagrams of 
configurations of density $\frac{1}{n}$ carry information on the global behaviours of 
$N$. This is notably due to the fact that configurations of density $\frac{1}{n}$ are unit
vectors and because of the underlying superposition principle for linear maps. That leads 
us to give the following proposition. 
\begin{proposition}
	\label{prop_densite-1} 
	Let $N$ be a $k$-\textsc{xor} circulant network of size $n$ with automata set $V$ 
	and global transition function $F$. The maximum convergence time, \emph{i.e.}, the 
	maximal transient trajectory length, is reached by configurations of density 
	$\frac{1}{n}$. Moreover, let $p_\ast$ be the period of the attractors reached by 
	configurations of density $\frac{1}{n}$. Then, for any configuration $x$ of $N$, 
	the period of its attractor divides $p_\ast$.
\end{proposition}
\begin{proof}
	All configurations of density $\frac{1}{n}$ are cyclic permutations of 
	each other. Thus, by Proposition~\ref{prop_basic}.4 their trajectories are isomorphic. 
	They consequently reach their attractor of period $p_\ast$ at the same time $t_\ast$. 
	Now, let $x$ be an arbitrary configuration and $i$ an automaton. By 
	Proposition~\ref{prop_symmetric-space-time-diagram}, the space-time diagram of 
	$(\overline{0}^{M_i(t)})_{t \in \mathbb{N}}$ is the reflected space-time diagram of 
	$(\overline{0}^i(t))_{t \in \mathbb{N}}$ with respect to $i$. Thus, the space-time 
	diagram of $(\overline{0}^{M_i(t)})_{t \in \mathbb{N}}$ reach its attractor at time 
	$t_\ast$ and its period is $p_\ast$. This means that, $\forall i \in N$, the trace of 
	automaton $i$ has period $p_\ast$ and enters its cyclic behaviour before 
	$t_\ast$. As a result, the trajectory of $x$ reaches its attractor before $t_\ast$ and 
	the period of the latter divides $p_\ast$.
\end{proof}

\subsubsection{$2$-\textsc{xor} circulant networks}
\label{sec_twoxor}

Now, we focus on $2$-\textsc{xor} circulant networks of arbitrary size $n$ 
and pay attention to the space-time diagrams of configurations of density 
$\frac{1}{n}$. What is called the \emph{interaction-step} of such a network $N$ is 
the smallest integer $s \neq 1 <n$ such that $\forall i \in V$, $(i, i+s) \in A$. 
Figure~\ref{fig_xor-ca-17-24} $(a)$ and $(b)$ illustrates as expected that when 
$s = 0$ the space-time diagram is the Sierpinski triangle. For other values of $s$, 
space-time diagrams are deformed Sierpinski triangles. These observations result in
the following lemma (that is used further to analyse $2$-\textsc{xor} circulant 
networks of size $n = 2^p$, $p \in \mathbb{N^*}$, and interaction step $s=0$).
\begin{lemma}
	\label{lem_local}
	If $N$ is a $2$-\textsc{xor} circulant network of size $n$ with
	interaction-step $s = 0$ then $\forall i \in V, \forall q \in \mathbb{N},\ x_i(2^q) 
	= x_{(i-2^q)}(0) \oplus x_i(0)$.
\end{lemma}
\begin{proof}
	Lemma~\ref{lem_local} is proven by induction on $q$. Let $i \in V$ be 
	an arbitrary automaton and let $q$ equal $1$ initially. Clearly, the following holds:
	\begin{equation*}
		x_i(2)\ =\ x_{(i-1)}(1) \oplus x_i(1)\ =\ x_{(i-2)}(0) \oplus x_{(i-1)}(0) \oplus 
		x_{(i-1)}(0) \oplus x_{i}(0)\ =\ x_{(i-2)}(0) \oplus x_{i}(0) \text{.}
	\end{equation*}
	Thus, the basis of the induction holds too. Now, consider that, for $q \in 
	\mathbb{N}$, $x_i(2^q) = x_{(i-2^q)}(0) \oplus x_i(0)$ is true. In the sequel, we 
	pay particular attention to states
	\begin{gather*}
		a = x_i(0)\text{,} \quad b = x_{(i-2^{q-1})}(0)\text{,}\quad c =
		x_{(i-2^q)}(0)\text{,} \quad d = x_i(2^{q-1})\text{,}\quad e = 
		x_{(i-2^{q-1})}(2^{q-1}) \quad \text{and} \quad f =
		x_i(2^{q})\text{.}
	\end{gather*}
	Then, by induction hypothesis, for $q+1$, we have $d = a \oplus b$, $e = b \oplus 
	c$ and $f = d \oplus e$, from which we derive that $f = d \oplus e = (a \oplus b) 
	\oplus (b \oplus c) = a \oplus c$. As a result, we can write:
	\begin{equation*}
		\forall i \in V, \forall q \in \mathcal{N},\ x_i(2^q) = x_i(0)
		\oplus x_{(i-2^q)}(0) \text{,}
	\end{equation*}
	and obtain the expected result.
\end{proof}

\subsubsection{$2$-\textsc{xor} circulant networks of sizes powers of $2$}

In this paragraph, we restrict the study to $2$-\textsc{xor} circulant 
networks of sizes $n = 2^p$, where $p \in \mathbb{N^*}$. Let $x = (x_0, \ldots, 
x_{n-1}) \in \{0,1\}^n$ be a configuration of such a network $N$. $x$ can be viewed 
as the concatenation of two vectors of sizes $\frac{n}{2}$ such that $x = (x',x'')$, 
where $x' = (x_0, \ldots, x_{\frac{n}{2}-1})$ and $x'' = (x_{\frac{n}{2}}, \ldots, 
x_{n-1})$. $x'$ and $x''$ are called the \emph{semi-configurations} of $x$. Let us 
define the \emph{repetition degree} $\delta_r(x)$ of $x$ as:
\begin{equation*}
	\delta_r(x = (x',x'')) = \begin{cases}
		0 & \text{if } x' \neq x'' \text{,}\\
		\delta & \text{if } (x' = x'') \land (\delta_r(x') = \delta - 1)\text{.} 
	\end{cases}
\end{equation*}
$x$ is said to be a \emph{repeated configuration} when $x = (x',x')$. Moreover, 
remark that the time complexity for computing the repetition degree is 
$\mathcal{O}(n)$. Let us now present results about such networks convergence 
times.
\begin{proposition}
	\label{prop_repeated-configuration}
	Let $N$ be a $2$-\textsc{xor} circulant network of size $n = 2^p$, $p \in
	\mathbb{N}^*$, and interaction-step $s$. Configurations $x \in \{0,1\}^n$
	of repetition degree $\delta_r(x) \geq \log_2(n) - 1$ converge towards $(0,
	\ldots, 0)$ in no more than $2$ time steps.
\end{proposition}
\begin{proof}
	First, notice that because $N$ is a $2$-\textsc{xor} circulant network 
	of size $n = 2^p$, $p \in \mathbb{N}^*$, there exist only $4$ repeated
	configurations of degree no smaller than $\log_2(n) - 1$, namely, $(0, 1,
	\ldots, 0, 1)$, its dual $(1, 0, \ldots, 1, 0)$ and $(1, \ldots, 1)$ and its
	dual $(0, \ldots, 0)$. Let us consider the two distinct parities of $s$
	independently. Also, let $t \in \mathcal{T}$ and let $x(t)$ be either
	$(0, 1, \ldots, 0, 1)$ or $(1, 0, \ldots, 1, 0)$. If $s$ is even, then, by hypothesis 
	on $x(t)$, $\forall i \in V,\ x_{(i+s)}(t+1) = x_i(t) \oplus x_{(i+s-1)}(t) = 1$. 
	Otherwise, if $s$ is odd, then, $\forall i \in V,\ x_{(i+s)}(t+1) = x_i(t) \oplus 
	x_{(i+s-1)}(t) = 0$. Now, considering with this Propositions~\ref{prop_basic}.2 
	and~\ref{prop_basic}.3, we get the expected result.
\end{proof}

\noindent Let us now focus on the particular case of $2$-\textsc{xor} circulant 
networks of sizes $n = 2^p$, $p \in \mathbb{N}^*$, and interaction-steps $s = 0$. 
\begin{theorem}
	\label{thm_convergence-lineaire}
	Let $N$ be a $2$-\textsc{xor} circulant network of size $n = 2^p$,  $p \in 
	\mathbb{N}^*$, and interaction-step $0$. The only attractor of $N$ is $(0, \dots, 0)$ 
	and any configuration $x$ converges to it in no more than $n$ time steps.
\end{theorem}
\begin{proof}
	Since $n = 2^p$, by Lemma~\ref{lem_local}, $\forall i \in V,\ x_i(n) = 
	x_i(0) \oplus x_{i+n}(0) = x_i(0) \oplus x_i(0) = 0$. Then, any configuration $x$ 
	converges to the stable configuration $(0, \dots, 0)$ in no more than $n$ time 
	steps.
\end{proof}

\noindent Questioning about the configurations whose convergence time is maximal 
leads us to Lemma~\ref{lem_repeated} and 
Theorem~\ref{thm_convergence-lineaire-impair}.
\begin{lemma}
	\label{lem_repeated}
	Let $N$ and $N'$ be two $2$-\textsc{xor} circulant networks of respective
	sizes $n = 2^{p+1}$ and $n' = 2^p$, $p \in \mathbb{N}^*$, and
	interaction-steps $0$. Let $x'$ be a configuration of size $2^p$ and $x = (x',
	x')$ be a repeated configuration of size $2^{p+1}$. Then, for any $t \in
	\mathcal{T}$, $x(t) = (x'(t), x'(t))$.
\end{lemma}
\begin{proof}
	Let $x$ be an arbitrary repeated configuration of $N$. By induction on 
	$t$, we show that $\forall t \in \mathbb{N},\ x(t) = (x'(t), x'(t))$. Let $G' = (V', 
	A')$ be the interaction graph of $N'$. By hypothesis, the lemma is true for $t = 0$. 
	Now, consider that $x(t) = (x'(t), x'(t))$ for $t \in \mathbb{N}$ ($x(t)$ is a repeated 
	configuration) and that $\forall i \in V,\ x_i(t + 1) = x_{(i-1)}(t) \oplus x_{i}(t)$. 
	Hence we have, for all $i \in V$:
	\begin{equation*}
		x_i(t + 1) = x_{(i-1)}(t) \oplus x_{i}(t) = x_{(i-1+2^p)}(t) \oplus x_{(i+2^p)}(t) = 
		x_{(i+2^{p})}(t+1) \text{.}
	\end{equation*}
	Consequently, $x(t+1)$ is also repeated and verifies, for all $i \in V'$:
	\begin{equation*}
		x_i(t+1) = x_{(i-1)~[n']}(t) \oplus x_{i}(t) = x'_{(i-1)~[n']}(t) \oplus x'_i(t) = 
		x'_i(t+1) \text{.}
	\end{equation*}
	As a result, $x(t+1) = (x'(t+1), x'(t+1))$.
\end{proof}
	
\begin{theorem}
	\label{thm_convergence-lineaire-impair}
	Let $N$ be a $2$-\textsc{xor} circulant network of size $n = 2^p$, $p \in 
	\mathbb{N}^*$, and interaction-step $0$. Any configuration $x$ such that $n \cdot 
	d(x) \equiv 1~[2]$ converges in $n$ time steps exactly.
\end{theorem}
\begin{proof}
	We proceed by induction on $p$. If $p = 1$, according to 
	Propositions~\ref{prop_basic}.3 and~\ref{prop_repeated-configuration}, 
	configurations of repetition degree $\log_2(n) - 1$ are proven to converge in $2$ 
	time steps. Thus, the basis of the induction holds. Consider the following induction 
	hypothesis: for $p = q$, any configuration $x$ such that $2^q \cdot d(x) \equiv 1~[2]$ 
	converges in $2^q$ time steps. Suppose now that $p = q + 1$ and consider a 
	$2$-\textsc{xor} circulant network $N$ of size $n = 2^{q+1}$ and interaction-step $0$. 
	Let $x$ be a configuration of size $2^{q+1}$ such that $n \cdot d(x) \equiv 1~[2]$. After 
	$2^q$ time steps:
	\begin{itemize}
	\item \emph{$x(2^q)$ is a repeated configuration of the form $x(2^q) =
		(x'(2^q),x'(2^q))$}. Indeed, by Lemma~\ref{lem_local}, $\forall i \in \{0, \ldots,
		2^q-1\},\ x_i(2^q) = x_i(0) \oplus x_{(i+2^{q})}(0)$. Hence, $\forall i \in \{0, 
		\ldots, 2^q-1\},\ x_i(2^q) = x_{(i+2^{q+1})}(0) \oplus x_{(i+2^{q})}(0) = 
		x_{(i+2^{q})}(2^q)$.
	\item \emph{$x'$ has an odd number of $1$s}. By the property above together with 
	Lemma~\ref{lem_local} , since $\forall i \in \{0, \ldots 2^q-1\}, x'_i(2^q) = 
	x_i(2^q) = x_i(0) + x_{(i+2^{q})}(0)$, each automaton of $x(0)$ influences exactly 
	one automaton of $x'$. If $x'_i(2^q) = 0$, then the states of both the automata of 
	$x(0)$ that influence $x'_i(2^q)$ must have the same parity. If $x'_i(2^q) = 1$ then 
	the states of both the automata of $x(0)$ that influence $x'_i(2^q)$ must have 
	different parities. Since there is an odd number of $1$s in $x(0)$, there is an odd 
	number of $1$s in $x'(2^q)$.
	\end{itemize}
	By Lemma~\ref{lem_repeated}, $x(2^q)$ behaves exactly like $x'(2^q)$. Furthermore, 
	by the induction hypothesis, $x'$ converges in exactly $2^q$ time steps. Hence, $x$
	converges in exactly $n = 2^{q+1}$ time steps.
\end{proof}

\section{Conclusion and Perspectives}
\label{sec_conclu}

In this paper, we have highlighted that non-monotony could be at the origin 
of dynamical singularities of Boolean models of genetic regulation networks, with 
respect to their sensitivity against synchronism. This is an interesting property 
because biological experimentations currently give no tangible results about the way
that genes express over time. Moreover, on the basis of this result, we have 
developed a study on the \textsc{xor} circulant networks class and have shown some 
notable results about their convergence times in particular.\medskip

This work opens many research directions that could help develop the
knowledge on the influence of non-monotony in automata networks and, \emph{a
  fortiori}, in real genetic networks. One of these perspectives relies on the
first part of this paper dealing with synchronism sensitivity. It would consist
in understanding how do monotone and non-monotone Boolean automata networks
relate. In~\cite{Noual2011a}, preliminary results are derived on synchronism
sensitivity of monotone networks that emphasise necessary structural conditions
(namely, the presence of specific circuits in the interaction graphs) and
examples of synchronism sensitive monotone networks are given. What is
interesting is that these examples involve linear \emph{monotone codings} of
non-monotony. This naturally raises the question of whether non-monotony
accounts for the synchronism sensitivity in arbitrary monotone and non-monotone
networks.  In addition, further analyses on the behaviours of \textsc{xor}
circulant networks are planned. We would like to obtain generalisations of the
results presented above by following two directions: relaxing structural
constraints step by step and viewing these networks as state transition systems
rather than discrete dynamical systems subjected to the parallel updating mode.

\section{Acknowledgements}
\label{sec_acknowledgments}

We are indebted to the \emph{Agence nationale de la recherche} and the 
\emph{R{\'e}seau national des syst{\`e}mes complexes} that have respectively 
supported this work through the projects Synbiotic (\textsc{anr} 2010 \textsc{blan} 
0307\,01) and M{\'e}t{\'e}ding (\textsc{rnsc} \textsc{ai}10/11-\textsc{l}03908).

%

%
\end{document}